\documentclass[letterpaper, 10 pt, conference]{ieeeconf}  

\IEEEoverridecommandlockouts                              
\usepackage{answers}
\usepackage{setspace}
\usepackage{graphicx}
\usepackage{multicol}
\usepackage{mathrsfs}
\usepackage{amsmath,amssymb}
\usepackage{bbold}
\usepackage{tikz}
\usetikzlibrary{backgrounds}
\usepackage[framemethod=tikz]{mdframed}
\usepackage{fancybox}

\usepackage[mathscr]{euscript}

\usepackage{schemabloc}

\usetikzlibrary{circuits}

\usepackage[english]{babel}
\usepackage[T1]{fontenc}

\usetikzlibrary{arrows, positioning, calc}

\usepackage{blox}

\usepackage{graphicx} 
\usepackage{epsfig} 
\usepackage{pgfplots}
\usepackage{pgfplotstable}
\usepackage{resizegather}
\usepackage{mdframed}

\mdfsetup{linewidth=0.75pt,
innertopmargin=.25pt,
  skipbelow=.25em,
  skipabove=.25em,
  nobreak=true
}

\usepackage{amsmath} 
\usepackage{amssymb}  
\usepackage{xcolor}
\usepackage{mathtools}
\usepackage{lipsum}
\usepackage{soul}
\usepackage{relsize}

\usepackage{amsthm}
\usepackage[ruled]{algorithm2e}
\LinesNumbered

\makeatletter
\newcommand{\removelatexerror}{\let\@latex@error\@gobble}
\makeatother
\allowdisplaybreaks[1]
\usepackage{balance}

\usepackage[normalem]{ulem}
\usepackage{xfrac}

\newtheorem{theorem}    {Theorem}
\newtheorem{prop} {Proposition}

\newtheorem{remark}     {Remark}

\theoremstyle{definition}

\newtheorem{assm} {Assumption}

\newcommand{\nc}{\newcommand}
\nc{\R}{{\mathbb R}}
\nc{\C}{{\mathbb C}}
\nc{\Z}{{\mathbb Z}}
\nc{\N}{{\mathbb N}}
\nc{\s}{\bar{\bf s}}
\nc{\I}{{\cal I}^*}
\nc{\e}{\boldsymbol{\mathsf e}}
\nc{\ellb}{\boldsymbol\ell}

\usepackage{pst-sigsys}

\let\labelindent\relax
\usepackage{enumitem}   

\usepackage{multicol,lipsum}

\usepackage{float}

\usepackage{nicematrix}

\usepackage{textcomp}

\usepackage{cite}

\usepackage{upgreek}

\usepackage{multibib}

\usepackage{hyperref} 
\hypersetup{
    plainpages=false,       
    unicode=false,          
    pdftoolbar=true,        
    pdfmenubar=true,        
    pdffitwindow=false,     
    pdfstartview={FitH},    
    pdftitle={Model Reference Adaptive Control with Linear-like Closed-loop Behavior},    
   pdfauthor={Mohamad T. Shahab and Daniel E. Miller},    
    pdfnewwindow=true,      
    colorlinks=true,        
    linkcolor=blue,         
    citecolor=black,        
    filecolor=magenta,      
    urlcolor=cyan          
}

\title{
\bf
Model Reference Adaptive Control
with 
Linear-like Closed-loop Behavior
}

\author{Mohamad T. Shahab and Daniel E. Miller
\thanks{M.T. Shahab 
        is with 
        the Computer, Electrical and Mathematical Science and Engineering (CEMSE) Division,
        King Abdullah University of Science and Technology 
        (KAUST), 
        Thuwal 23955, Saudi Arabia.
        Email: {\tt mohamad.shahab@kaust.edu.sa}.
        }
\thanks{
D.E. Miller is with the Department of Electrical and Computer Engineering, 
        University of Waterloo, Waterloo, ON N2L 3G1, Canada.
        Email: {\tt miller@uwaterloo.ca}.
        }
\thanks{
Support for this work was provided by the Natural
Sciences and Engineering Research Council of Canada
(NSERC).
}
        }

\begin{document}

\baselineskip=9.745pt
\maketitle
\thispagestyle{empty}
\pagestyle{empty}

\begin{abstract}
It is typically proven in adaptive control that asymptotic stabilization and tracking holds, and that at best a bounded-noise bounded-state property is proven.
Recently, it has been shown in both the pole-placement control and the $d$-step ahead control settings that 
if, as part of the adaptive controller, a
parameter estimator based on the original projection algorithm is used
and the parameter estimates 
are restricted 
to a convex
set, 
then
the closed-loop system experiences linear-like behavior:
exponential stability, a bounded gain on the noise in every $p$-norm, and a convolution bound on the exogenous inputs;
this can be leveraged to provide tolerance to unmodelled dynamics and plant parameter time-variation.
In this paper, we extend 
the approach to
the more general Model Reference Adaptive Control (MRAC)
problem and demonstrate that
we achieve the same desirable linear-like closed-loop properties.

\end{abstract}

\section{Introduction}

Adaptive control is an approach used to deal with systems 
with uncertain and/or time-varying 
parameters. In the classical approach to adaptive control, one combines a
linear time-invariant (LTI) compensator together with a tuning mechanism to adjust
the compensator parameters to match the plant.
The first general proofs came around 1980, e.g. see \cite{morse1978}, \cite{Goodwin1980}, \cite{Morse1980}, \cite{Narendra1980} and \cite{Narendra1980_pt2}.
However, the original controllers
are typically not robust to unmodelled dynamics, do not tolerate time-variations well,
have poor transient behavior and do not handle noise/disturbances well, e.g.
see \cite{rohrs}. 
During the following two decades, a good deal of research was carried out
to alleviate these shortcomings; a number of small controller design changes were proposed, such as the use of signal normalization, 
deadzones and $\sigma$-modification, e.g.
see \cite{Ioa86}, \cite{kreiss}, \cite{rick2}, \cite{rick}, 
and \cite{Tsakalis4};
also, simply using projection onto a convex set
of admissible parameters turned out to be powerful, 
e.g. see \cite{hanfu}, \cite{Naik}, \cite{Wen}, \cite{Wenhill} and \cite{ydstie}. 
However, in general these redesigned
controllers provide asymptotic stability 
and not exponential stability, with no bounded
gain on the noise\footnote{An exception is the work 
of Ydstie \cite{ydstie} where a bounded gain is proven.}; 
that being said, some of them, especially those using projection,
provide a bounded-noise bounded-state property, 
as well as tolerance of some degree
of unmodelled dynamics and/or time-variations.

Recently, 
for
discrete-time
LTI plants,
in both the $d$-step ahead control setting \cite{scl17}, \cite{acc19}, \cite{mcss20},
and the pole-placement control setting \cite{ccta17}, \cite{mcss18},
a new approach has been proposed  
which not only provides exponential
stability and a bounded gain on the noise, but also a
convolution bound on the exogenous inputs;
the resulting
convolution bound is leveraged to prove 
tolerance 
to a degree of time-variations and to a degree of unmodelled
dynamics \cite{ccta20}. 
As far as the authors are aware, such
{\bf linear-like convolution bounds have never before been
proven in the adaptive setting}. 
The key idea is to use the
original projection algorithm in conjunction with a restriction
of the parameter estimates to a convex set,
although
this convexity requirement was relaxed in
\cite{cdc18}, \cite{mcss18} and \cite{tac20}.
The goal of the present paper is
to extend these linear-like results
in the $d$-step ahead control setting to
the more general
{\em Model Reference Adaptive Control} (MRAC) problem.

Model reference adaptive control 
is an important approach to adaptive control 
where a pre-designed stable reference model is
used to model the desired closed-loop behavior.
Here we build on the results
proven for the $d $-step-ahead adaptive control problem
in \cite{acc19} and \cite{mcss20},
which is a special case of the more general MRAC problem considered here;
because we are seeking stronger closed-loop properties than what
is normally proven in the literature, more detailed analysis is needed in dealing with the MRAC setup,
since
the introduction of
the reference model into the analysis brings
extra complexity.
We prove that the
desirable linear-like closed-loop properties of
{\bf exponential stability, a bounded gain on the noise in every $p$-norm and
a convolution bound on the exogenous inputs},
are achieved using a model reference adaptive controller;
we also prove a stronger tracking result than what is usually found in the literature.

{\bf Notation.} We use standard notation throughout the paper. We denote $\R$, $\Z$, and $\C$ as the set of real numbers, integers and complex numbers, respectively.
  We will denote the Euclidean-norm of a vector and the induced norm of a matrix by 
the subscript-less default notation $\|\cdot\|$. 
Let
${\mathbb S}(\R^{p\times q} ) $ denote
the set of $\R^{p\times q} $-valued sequences. 
Also, $\ellb_{\infty} $ denotes the set of bounded sequences.
For a signal $f\in\ellb_\infty $, define the $\infty $-norm by
$\|f \|_\infty:= \sup_{t\in \Z} | f(t) | $.
For a closed and convex set $\Omega\subset\R^p$, let the function
 $\mathrm{Proj}_{\Omega}
 \left\{ \cdot \right\}:\R^{p}
 \mapsto \Omega $ 
denote the projection 
onto the set $\Omega  $;
because the set $\Omega $ is closed and convex, 
the function $\mathrm{Proj}_{\Omega} $ is well-defined.
If $\Omega\subset\R^p$ is a compact (closed and bounded) set, 
we define $\|\Omega\|:=\max_{x\in\Omega}\|x\|$. 
Let $I_{p} $ denote the
identity matrix of size $p$.
Define the normal vector $\e_j \in\R^{p}$ 
of appropriate length $p$ as
  $\e_j:=
    \bigl[
  \underbrace{\begin{matrix}0 & \cdots & 0 \end{matrix}}_{j-1\text{ elements} } \;\; \begin{matrix}1 & 0 & \cdots & 0 \end{matrix}
    \bigr]^\top.$
Last of all, for a signal
$f\in {\mathbb S}(\R ) $ which is 
sufficiently well-behaved to have a $z$-transform,
we let $F(z) $ denote this quantity.

 \section{The Setup}
\label{sec2}

In this paper we consider the following linear time-invariant (LTI)
discrete-time plant:
\begin{equation} \label{plant1}
\sum_{i=0}^{n} a_i y(t-i) = \sum_{i=0}^{m} b_i u(t-d-i) +w(t), \quad t\in\Z,
\end{equation}
with
 $y(t)\in\R $ as the measured output,
$u(t)\in\R $ as the control input, 
and $w(t)\in\R $ as the noise/disturbance.
The plant parameters are regularized such that $a_0=1 $,
and the system delay is exactly $d$, i.e. $b_0\neq0 $.
Associated with this plant are the polynomials
${\mathbf A}(z^{-1}):= \sum_{i=0}^{n} a_i z^{-i}
$ and  $
\quad {\mathbf B}(z^{-1}):= \sum_{i=0}^{m} b_i z^{-i}$,
the transfer function $z^{-d} \frac{{\mathbf B}(z^{-1})}{{\mathbf A}(z^{-1})} $,
and the plant parameter vector:
$$ \theta:= \begin{bmatrix}
a_1 & a_2 & \cdots & a_n & b_0 & b_1& \cdots & b_m
\end{bmatrix};
 $$
we assume that $\theta $ belongs to a known set $ {\cal S}_{ab} \subset\R^{n+m+1}$.
Observe that such a plant can be expressed in the $z$-transform domain as
\begin{equation} \label{plant2}
{\mathbf A}(z^{-1}) Y(z) = z^{-d} {\mathbf B}(z^{-1}) U(z) + W(z).
\end{equation}

The control objective is closed-loop stability and asymptotic tracking 
of a given reference signal $y^*(t)\in\R $ generated as the output of a stable reference model;
more specifically,
given pre-designed polynomials
${\mathbf H}(z^{-1}):= \sum_{i=0}^{n'-d} h_i z^{-i}$ and $ 
{\mathbf L}(z^{-1}):= 1+\sum_{i=1}^{n'} l_i z^{-i}$ 
(with $ n'\leq n $),
and given
a bounded exogenous signal $r(t)\in\R $, 
we utilize the following reference model expressed in the $z $-transform form:
\begin{flalign}
 {\mathbf L}(z^{-1})Y^*(z)=z^{-d}{\mathbf H}(z^{-1}) R(z).
 \label{plantRef1}
 \end{flalign}
We assume that the roots of ${\mathbf L}(z^{-1})$ belongs to the open unit desk,
i.e. the reference model is stable.
If we define the tracking error $\varepsilon$ by
\begin{flalign}
 \varepsilon(t):=y(t)-y^*(t),
\end{flalign}
then the goal 
is to drive $\varepsilon $ to zero asymptotically.
\begin{remark}
Notice that for the $d$-step-ahead control problem
the reference model is simply $Y^*(z)=z^{-d} R(z) $.
\end{remark}

We impose the following assumptions on the set of admissible parameters.
\begin{assm}
${\cal S}_{ab}$ is closed and bounded (compact),
and
for each $\theta\in {\cal S}_{ab}$,
the corresponding ${\mathbf B}(z^{-1})$ has roots in the open unit disk and
the sign of $b_{0}$ is always the same.
\end{assm}
\noindent
The boundedness requirement on ${\cal S}_{ab}$ is reasonable in practical situations;
it is used here to prove uniform bounds and decay rates on the closed-loop behavior.
The constraint on the roots of ${\mathbf B}(z^{-1})$ is a requirement that the plant 
be minimum phase;
this is necessary to ensure tracking of 
bounded reference signals \cite{ld_paper}.
Knowledge of the sign of the high-frequency gain $b_0 $ is common 
in adaptive control \cite{goodwinsin}.
\begin{remark}
It is implicit in the assumptions that we know
the system delay $d$
as well as the upper bounds on the orders 
of ${\mathbf A}(z^{-1})$ and ${\mathbf B}(z^{-1})$.
\end{remark}


To proceed, we use a parameter estimator 
together with an adaptive control law based on the Certainty Equivalence Principle.
It is convenient to put the plant into the so-called {\em predictor form}.
To this end,
by long division 
we can find ${\mathbf F}(z^{-1})=\sum_{i=0}^{d-1} f_i z^{-i} $ and ${\boldsymbol\alpha}(z^{-1})=\sum_{i=0}^{n-1} \alpha_i z^{-i}$ that satisfy the following:
\begin{equation*}
\frac{{\mathbf L}(z^{-1})}{{\mathbf A}(z^{-1})} = {\mathbf F}(z^{-1}) 
+ z^{-d} \frac{{\boldsymbol\alpha}(z^{-1})}{ {\mathbf A}(z^{-1})};
\end{equation*}
if we now define
$
{\boldsymbol\beta}(z^{-1})= \sum_{i=0}^{m+d-1} \beta_i z^{-i}:={\mathbf F}(z^{-1}){\mathbf B}(z^{-1}),
$
then it is easy to verify that the following is true:
\begin{equation} \label{alpha1}
z^{-d} \frac{{\mathbf B}(z^{-1})}{{\mathbf A}(z^{-1})} = \frac{ {\boldsymbol\beta}(z^{-1})}{z^{d}
{\mathbf L}(z^{-1})-{\boldsymbol\alpha}(z^{-1})} .
\end{equation}
So comparing \eqref{alpha1} with the plant equation in \eqref{plant2}, 
we are able to re-write the plant equation as
\begin{equation} \label{plantM1}
{\mathbf L}(z^{-1})[z^d Y(z)] = {\boldsymbol\alpha}(z^{-1})Y(z) + {\boldsymbol\beta}(z^{-1}) U(z) 
+ \overline W(z),
\end{equation}
with $ \overline W(z):= z^d {\mathbf F}(z^{-1}) W(z) $.
Now define
a weighted sum of the system output $\overline y $ by
\begin{flalign}
\overline y(t):= y(t) + \sum_{j=1}^{n'} l_j y(t-j);
\label{ybar}
\end{flalign}
clearly the $z$-transform 
of $\overline y(t) $ is
$ 
{\mathbf L}(z^{-1})Y(z) $,
so the time-domain counterpart of \eqref{plantM1}
in predictor form is
\begin{equation} \label{plantM2}
\overline y(t+d) = \phi(t)^\top \theta^* + \overline w(t),
\end{equation}
with
\[
	\phi(t):=
	\begin{bmatrix}
	y(t) \\ y(t-1) \\ \vdots \\ y(t-n+1) \\ u(t) \\ u(t-1) \\ \vdots \\ u(t-m-d+1)
	\end{bmatrix},
	\qquad
	\theta^*:=
	\begin{bmatrix}
	\alpha_0 \\ \alpha_1 \\ \vdots \\ \alpha_{n-1} \\
	 \beta_0 \\ \beta_1 \\ \vdots \\ \beta_{m+d-1}
	\end{bmatrix}.
\]

Let ${\cal S}_{\alpha\beta} \subset \R^{n+m+d} $ denote the set of admissible $\theta^* $
that arise from the original plant parameters which lie in ${\cal S}_{ab} $;
it is clear that the associated mapping from ${\cal S}_{ab} $ to ${\cal S}_{\alpha\beta}$
is analytic, so the compactness of ${\cal S}_{ab} $ means that ${\cal S}_{\alpha\beta}$ is compact as well.
Furthermore, it is easy to see that $\beta_0=b_0 $.
It is convenient that the set of admissible parameters in the new parameter space be convex and closed;
so at this point let ${\cal S} \subset \R^{n+m+d} $ be any compact and convex set 
containing ${\cal S}_{\alpha\beta}$ for which the $(n+1)$th element is never zero; the convex hull of ${\cal S}_{\alpha\beta}$ will do,
although it may be more convenient to use a 
hyperrectangle (for projection purposes).
We will show an example on obtaining such a set in the simulation section.

Now define
$\overline Y^*(z):= {\mathbf L}(z^{-1})Y^*(z); $
then the model reference control law is 
given by
\begin{equation*} 
\overline y^*(t+d) = \phi(t)^\top \theta^*.
\end{equation*}
In the absence of noise, and assuming the controller is applied for all
$t\in\Z $, 
we can show that
we have $y(t)=y^*(t) $ for all $t\in\Z $.
In our case of unknown parameters, we seek an adaptive version of the control law 
which is applied after some initial time, 
i.e. $t\geq t_0 $.

\subsection{Initialization}

In most adaptive control results, the goal is to prove asymptotic behavior, so the
details of the initial condition are unimportant. 
On the other hand, here we wish to obtain a bound
on the transient behavior so we must proceed carefully. 
Here we adopt the approach used in
the $d$-step ahead control setting of \cite{acc19} and \cite{mcss20}.
With the definition \eqref{ybar} in mind (and $n'\leq n $), 
observe that if we wish to solve \eqref{plantM2} 
for $y(\cdot)$ starting at time $t_0$, then it
is clear that we need an initial condition of
\begin{multline}
x_0:=
\bigl[
y(t_0) \;\; 
y(t_0-1) \;\;
\cdots \;\; y(t_0-n-d+2) 
\nonumber
\\
\qquad
u(t_0) \;\; 
u(t_0-1) \;\;
\cdots \;\; u(t_0-m-2d+2) 
\bigr]^\top
;
\end{multline}
observe that this is sufficient information to obtain 
$\phi(t_0),\phi(t_0-1),\ldots,\phi(t_0-d+1) $.

 \subsection{Parameter Estimation}
\unskip
We can re-write the plant equation \eqref{plantM2} as
\begin{equation} \label{plantM3}
\overline y(t+1) = \phi(t-d+1)^\top \theta^* + \overline w(t-d+1),
\;\; t\geq t_0.
\end{equation}
Given an estimate $\hat\theta(t) $ of $\theta^* $ at time $t$,
we define the prediction error by
\begin{equation} \label{predict1}
e(t+1):=\overline y(t+1) - \phi(t-d+1)^\top \hat\theta(t);
\end{equation}
this is a measure of the error in $\hat\theta(t) $.
A common way to obtain a new estimate is from the solution of the optimization problem
\[
\underset{\theta}{\mathrm{argmin}} \left\{
\|\theta-\hat\theta(t) \| : 
\overline y(t+1)=\phi(t-d+1)^\top \theta
\right\},
\]
yielding the (ideal) {\bf Projection Algorithm}:
\begin{equation}
\hat\theta(t+1)
 =
 \left\{
 \begin{matrix*}[l]
 \hat\theta(t)
 & \phi(t-d+1)=0 
 \\
 \hat\theta(t)
 +
 \frac{\phi(t-d+1)}{\|\phi(t-d+1)\|^2}
  e(t+1)
  & \text{otherwise};
  \end{matrix*}
  \right.
  \label{orig1}
\end{equation}
at this point, we can also constrain it to ${\cal S}$ by projection. Of
course, if $\|\phi(t-d+1)\|$ is close to zero, numerical problems
may occur, so it is the norm in the literature (e.g. \cite{goodwinsin} and \cite{Goodwin1980})
to add a constant to the denominator\footnote{
An exception is \cite{akhtar} where
  the ideal algorithm \eqref{orig1} is used and
  Lyapunov stability is proven, 
  but a convolution bound on the exogenous inputs is not proven, and the high-frequency gain is assumed to be known.
}; 
however as pointed out in our earlier work
\cite{ccta17}, \cite{mcss18} and \cite{mcss20}, 
this can lead to the loss of exponential
stability and a loss of a bounded gain on the noise.
As proposed in \cite{ccta17}, \cite{mcss18} and \cite{mcss20}, 
we turn off the estimation if it is clear that the noise
is swamping the estimation error. To this end,
with $\delta\in(0,\infty] $, we turn off the estimator if the 
update is
larger than $2\|{\cal S}\| + \delta$ in magnitude;
so
define
\begin{equation}
\rho(t)
:=\left\{\begin{matrix}
1 & & \text{if } | e(t+1)|< (2\Vert {\cal S}\Vert +\delta)\|\phi(t-d+1)\|\\
0 & & \text{otherwise};
\end{matrix}\right.
\nonumber
 \end{equation}
given the initial condition of $\hat\theta(t_0 )
=\theta_0 \in \R^{m+n+d}  $,
 for $t\geq t_0 $ we define\footnote{If $\delta=\infty $, then we adopt the understanding that $\infty \times 0 = 0$, in which case this formula in \eqref{est1a} collapses into
the original version \eqref{orig1}.}
\begin{subequations}
\label{est1}
\begin{flalign}
\label{est1a}
 &\check\theta(t+1)
 =
 \hat\theta(t)
 +
 \rho(t)
  {\frac{\phi(t-d+1)}{\|\phi(t-d+1)\|^2}}
  e(t+1)
 \\
\label{est1b}
 &\hat\theta(t+1)
 =
\mathrm{Proj}_{\cal S} \left\{ \check\theta(t+1) \right\}.
\end{flalign}
\end{subequations}

Analyzing the closed-loop system requires a careful examination of the estimation algorithm.
First define the parameter error 
by $\tilde\theta(t) := \hat\theta(t)-\theta^* $.
The following result lists properties which are 
equivalent to those
of Proposition 1 in \cite{mcss20} 
for the
$d $-step ahead adaptive control setup.
\begin{prop}
\label{est_prop}
For every 
 $t_0\in\Z$, $ x_0\in\R^{n+m+3d-2}$, $
\theta_0 \in{\cal S}$, $ \theta\in{\cal S}_{ab}
, w\in\ellb_\infty
 $,
and $\delta\in(0,\infty] $,
when the estimator \eqref{est1} is applied to the plant \eqref{plant1},
the following holds:
\begin{flalign}
&\|\hat\theta(t+1)-\hat\theta(t) \|
\leq
\rho(t)
\frac{|e(t+1) |}{\|\phi(t-d+1) \|},
\quad t\geq t_0,
\nonumber
\\
&\|\tilde\theta(t) \|^2
\leq
\|\tilde\theta(\tau) \|^2
+
\sum_{j=\tau}^{t-1}
\rho(j)
\biggl[
- \frac{1}{2}\frac{e(j+1) ^2 }{ \|\phi(j-d+1) \|^2 }
+
\nonumber
\\
&\qquad\qquad\qquad
\frac{2\overline w(j-d+1)^2 }{ \|\phi(j-d+1) \|^2 }
\biggr],
\qquad
t>\tau\geq t_0.
\nonumber
\end{flalign}
\end{prop}

\subsection{The Control Law}

With the natural partitioning 
\begin{equation}
	\hat\theta(t)=:
	\begin{bmatrix}
	\hat\alpha_0(t) 
	& \cdots & \hat\alpha_{n-1}(t) &
	 \hat\beta_0(t) 
	 & \cdots & \hat\beta_{m+d-1}(t)
	\end{bmatrix}^\top,
	\nonumber
\end{equation}
the {\bf model reference adaptive control law}
(based on the Certainty Equivalence principle) 
is
\begin{equation} 
\overline y^*(t+d) = \phi(t)^\top \hat\theta(t);
\nonumber
\end{equation}
solving this for $u(t)$
and using the reference model \eqref{plantRef1},
we have
\begin{flalign}
u(t)
&=
\frac{1}{\hat\beta_0(t)}\biggl[
-\sum_{i=0}^{n-1} \hat\alpha_{i}(t) y(t-i)
-\sum_{i=1}^{m+d-1} \hat\beta_{i}(t) u(t-i) +
\nonumber
\\
&\qquad
 \sum_{i=0}^{n'-d} h_i r(t-i)
\biggr],
\quad t\geq t_0.
 \label{control1}
\end{flalign}

It is convenient for analysis to 
define an auxiliary
tracking error:
\begin{flalign}
\overline \varepsilon(t):=\overline y(t)-\overline y^*(t);
\label{errL}
\end{flalign}
it is easy to show that
\begin{flalign}
\overline\varepsilon(t)
&=
 -\phi(t-d)^\top \tilde\theta(t-d) + \overline w(t-d),
 \;\; t\geq t_0+d,
\label{error2}
\\
e(t)&= 
-\phi(t-d)^\top \tilde\theta(t-1) + \overline w(t-d),
\;\; t\geq t_0+1,
\label{pred_error1}
\end{flalign}
as well as
\begin{equation} \label{error1}
\overline\varepsilon(t)=e(t)+ \phi(t-d)^\top \left[\hat\theta(t-1)-\hat\theta(t-d)\right],
\quad t\geq t_0+d.
\end{equation}
Observe that we can compute
$\overline \varepsilon (t),\, t\in\{t_0,t_0+1,\ldots,t_0+d-1 \} $, from $x_0,w$ and $y^* $.

In the next section we develop several models used in the analysis, after which we state and prove
our result.
The approach borrows ideas from our previous work
on the $d$-step ahead setup \cite{mcss20}
and extends them to the
{\em Model Reference Adaptive Control} (MRAC) case.


\section{The Analysis}
\unskip
In the pole-placement adaptive control setup of our earlier work \cite{mcss18},
a key closed-loop model consists of an update equation for $\phi(t) $, with the state
matrix consisting of controller and plant estimates; 
this was effective because the characteristic polynomial of this matrix is time-invariant 
and has all roots in the open unit disk.
If we were to apply the same idea in our case here, then the characteristic polynomial 
would have roots which are time-varying, with some at zero and the rest at the roots
of the corresponding naturally defined polynomial $\hat{\boldsymbol \beta}(t,z^{-1} ) $,
which is time varying, and it may not have roots in the open unit disk.
On the other hand, in the $d$-step ahead adaptive control setup of our earlier work \cite{acc19} and \cite{mcss20},
these difficulties
were dealt with by constructing
three different models for use in the analysis:
a model that does not use parameter estimates but is driven by the tracking error,
a crude model to bound the size of growth of $\phi(t) $,
and a crucial model which is driven by perturbed versions of the present and past values of $\phi(\cdot) $.
Here in this paper, which deals with the more general
MRAC problem,
 we construct similar, though not identical, 
 models,
but they need more careful analysis than ones in the $d $-step ahead control case.

\subsection{A Good Model}
\unskip
Here we first obtain an equation which avoids using
parameter estimates, though it is driven by the weighted sum of the tracking error $\overline\varepsilon(\cdot) $.
By extending 
the idea from \cite{acc19} and \cite{mcss20}, 
using the definition of $\varepsilon $ we obtain
a formula for $y(t+1) $,
and 
using the plant equation \eqref{plant1} we obtain
a formula for $u(t+1) $;
then, 
it is easy to see that there exists a matrix $A_g\in\R^{(n+m+d)\times(n+m+d)} $
(which depends implicitly on $\theta \in {\cal S}_{ab} $) so that the following holds:
\begin{flalign} 
&\phi(t+1)
=
A_g \phi(t) +
  \e_1
 \varepsilon(t+1) +
\nonumber \\ 
&\quad 
\frac{1}{b_0}
\e_{n+1}
\sum_{i=0}^{d} a_{d-i} \varepsilon(t+1+i) 
+ 
\e_1 y^*(t+1)
+
\nonumber \\ 
&\quad
\frac{1}{b_0}
\e_{n+1}
\biggl[
\sum_{i=0}^{d} a_{d-i} y^*(t+1+i) - w(t+d+1)
\biggr].
\label{plant_good1}
\end{flalign}
The characteristic polynomial of $A_g $ 
is $\tfrac{1}{b_0}z^{n+m+d}{\mathbf B}(z^{-1}) $, 
so all of its roots are in the open unit disk.

The model in \eqref{plant_good1} is similar to the {\em good} model obtained in 
the analysis in the $d $-step ahead control case in \cite{acc19} and \cite{mcss20}
where it is driven by the tracking error $\varepsilon $.
However in the case considered here we would like to obtain 
a model which is, instead, driven by $\overline\varepsilon $;
this will turn out to be crucial in analyzing the closed-loop behavior.
To this end,
from \eqref{errL} and the definitions of $\overline y $ and $\overline y^* $,
it is easy to see that 
\begin{equation}
{\cal E}(z)=\frac{1}{{\mathbf L}(z^{-1})} \overline{\cal E}(z)
;
\label{tfL}
\end{equation}
so
we can represent $\varepsilon(t) $ as the output of an 
$n' $th-order system
driven by
$\overline \varepsilon $ as follows: 
with 
$
\zeta(t)
:=
\begin{bmatrix}
\varepsilon(t) & \varepsilon(t-1) & \cdots & \varepsilon(t-n'+1)
\end{bmatrix}^\top,
$
and
$A_l\in\R^{n'\times n' } $ defined by
\[
A_l
:=\left[
\begin{smallmatrix}
-l_1 & -l_2 & \cdots & -l_{n'-1} & -l_{n'}
\\
1 & 0 & \cdots & 0 & 0
\\
0 & 1 & \cdots & 0 & 0
\\
\vdots & & \ddots & \ddots & \vdots
\\
0 & 0 & \cdots & 1 & 0
\end{smallmatrix}
\right],
\]
we have
\begin{subequations}
\label{xi_sys}
\begin{flalign}
\zeta(t+1) 
&= 
A_l \zeta(t) + 
\e_1 \overline\varepsilon(t+1)
\\
\varepsilon(t)
&=
\e_1^\top \zeta(t).
\label{vare_eqn}
\end{flalign}
\end{subequations}
Note that \eqref{plant_good1} is driven on the RHS by $d+1 $ terms of $\varepsilon(\cdot) $;
but from \eqref{vare_eqn} we have
\begin{equation}
\varepsilon(t+1+j)=\e_1^\top \zeta(t+1+j ),\quad j=0,1,\ldots,d. 
\label{barerr2}
\end{equation}
With this in mind,
we construct the following $(n'(d+1)) $th-order system
driven by $\overline \varepsilon(\cdot) $:
\begin{equation}
\left[
\begin{smallmatrix}
\zeta(t+d+2)
\\
\zeta(t+d+1)
\\
\vdots
\\
\vdots
\\
\zeta(t+2)
\end{smallmatrix}
\right]
= 
\underbrace{
\left[
\begin{smallmatrix}
A_l &   &   &   &  
\\
I_{n'} &  &  &  & 
\\
 & I_{n'} &   &  & 
\\
 & & \ddots &  & 
\\
 &  &  & I_{n'} & 0
\end{smallmatrix}
\right]
}_
{=:\widetilde A_l }
\underbrace{
\left[
\begin{smallmatrix}
\zeta(t+d+1)
\\
\zeta(t+d)
\\
\vdots
\\
\vdots
\\
\zeta(t+1)
\end{smallmatrix}
\right]
}_{=:\overline\zeta(t)}
+
\e_1
\overline \varepsilon(t+d+2)
.
\label{xi_sys_no2}
\end{equation}
At this point we can combine the models
\eqref{plant_good1} and \eqref{xi_sys_no2} 
together with the linking equation 
\eqref{barerr2}
to obtain a model driven by the exogenous inputs
and $\overline \varepsilon $
(rather than $\varepsilon $):
with 
\begin{equation}
\eta(t)
:=
\frac{1}{b_0}
\e_{n+1}
\biggl[
\sum_{i=0}^{d} a_{d-i} y^*(t+1+i) - w(t+d+1)
\biggr]
+
\e_1 y^*(t+1),
\label{eta_2}
\end{equation}
if follows that there exists
a matrix
 $\tilde B\in\R^{(n+m+d)\times(n'(d+1))}  $, which depends continuously on $\theta \in {\cal S}_{ab}  $,
to obtain
the following $(n+m+d+n'(d+1)) $th-order system:
\begin{equation}
\begin{bmatrix}
\phi(t+1)
\\
\overline\zeta(t+1)
\end{bmatrix}
=
\underbrace{
\begin{bmatrix}
A_g & \tilde B
\\
 & \tilde A_l
\end{bmatrix}
}_{=: \widetilde A_g }
\underbrace{
\begin{bmatrix}
\phi(t)
\\
\overline\zeta(t)
\end{bmatrix}
}_{=: \overline\phi(t) }
+ \eta(t)
+
\e_{n+m+d+1}
\overline\varepsilon(t+d+2).
\label{goodmodel3}
\end{equation}

Before presenting an even 
better model suitable for analysis 
we need to analyze a couple of crude models of the closed-loop behavior.

\subsection{Crude Models}
\unskip
At times, we will need to use crude models to bound the size of the growth of
$\phi(t) $ 
and
the size of the growth of
$\overline\phi(t) $ 
in terms of the exogenous inputs. Following a similar analysis
of the crude model in \cite{acc19} and \cite{mcss20},
we now use \eqref{plant1} to describe $y(t+1)$,
and use the control law \eqref{control1}
together
with the obtained equation for $y(t+1) $
to describe $u(t+1)$;
so, we can appropriately define matrices $A_{1}(t)$, $B_{1}(t)$ and $B_{2}(t)$ in terms of $\theta\in{\cal S}_{ab} $ 
and $\hat\theta(t+1)\in{\cal S} $
so that 
we have the following {\bf crude model of the behavior of $\phi(\cdot) $}:
\begin{flalign}
&\phi(t+1)=A_1(t)\phi(t)+
\nonumber
\\
&
\quad
B_1(t) \overline y^*(t+d+1) +B_2(t) w(t+1),
\quad
t\geq t_0.
\label{crude1}
\end{flalign}
Furthermore,
we combine \eqref{crude1} and \eqref{xi_sys_no2} 
to obtain an equation for $\overline\phi(t) $:
we can appropriately define matrices $B_3(t),B_4(t) $ so that
\begin{flalign}
&\overline\phi(t+1)
=
\begin{bmatrix}
A_1(t) & 
\\
 & \widetilde A_l
\end{bmatrix}
\overline \phi(t)
+B_3(t) \overline y^*(t+d+1) +
\nonumber
\\
&\quad
B_4(t) w(t+1)
+
\e_{n+m+d+1} \overline\varepsilon(t+d+2),
\quad 
t\geq t_0.
\label{crude_no2}
\end{flalign}
Now we want to find a representation for $\overline\varepsilon(t+d+2)$ on the RHS above
in terms of $\phi(t) $:
from \eqref{error2} 
we have $\overline \varepsilon(t+d+2)=-\tilde\theta(t+2)^\top\phi(t+2)+\overline w(t+2)   $,
so we use \eqref{crude1} 
to find a representation of $\phi(t+2) $ in terms of $\phi(t) $
and
substitute 
into \eqref{crude_no2}; then we can appropriately 
define matrices $A_2(t), B_5(t), B_6(t), B_7(t), B_8(t) $ so that the following {\bf crude model of the behavior of $\overline\phi(\cdot) $} is obtained:
\begin{flalign}
&\overline\phi(t+1)
=
A_2(t)
\overline \phi(t)
+B_5(t) \overline y^*(t+d+1) +
\nonumber
\\
&\qquad
B_6(t) w(t+1)
+
B_7(t) \overline y^*(t+d+2) +
\nonumber
\\
&\qquad
B_8(t) w(t+2)
+
\e_{n+m+d+1} \overline w(t+2),
\quad
t\geq t_0.
\label{crude_no3}
\end{flalign}
Due to the compactness of ${\cal S}_{ab} $, ${\cal S}_{\alpha\beta} $ and ${\cal S} $,
we can obtain the following immediately.
\begin{prop}
\label{prop_crude}
There exists a constant $c_1\geq 1 $ such that
for every $t_0\in\Z $, $x_0\in\R^{n+m+3d-2} $, $\theta_0 \in{\cal S} $,
$\theta\in{\cal S}_{ab} $, $r,w\in\ellb_\infty $, and $\delta\in(0,\infty] $,
when the adaptive controller \eqref{est1} and \eqref{control1} is applied to the plant \eqref{plant1},
the following holds:
\begin{flalign}
&\|A_1(t) \|\leq c_1, \;\;  \|A_2(t) \|\leq c_1, \;\; 
\|B_1(t) \|\leq c_1, 
\nonumber
\\
&   \|B_2(t) \|\leq c_1, \;\;
\|B_5(t) \|\leq c_1, \;\; 
\|B_6(t) \|\leq c_1, 
\nonumber
\\
& 
\|B_7(t) \|\leq c_1, \;\; 
\|B_8(t) \|\leq c_1,
\;\;
t\geq t_0.
\nonumber
\end{flalign}
\end{prop}

\subsection{A Better Model}
\unskip
The {\em good} closed-loop model \eqref{goodmodel3} is driven by 
a future value of $\overline \varepsilon(\cdot) $.
We now combine it with the crude model \eqref{crude_no3} to obtain
a new model which is driven by perturbed version of $\overline\phi(t) $,
with weights associated with the parameter estimation updates. 
Before proceeding,
motivated by the form of the term in the
parameter estimator, 
we define
\begin{flalign}
\nu(t):=
\rho(t)
\frac{\phi(t-d+1)}{\Vert\phi(t-d+1)\Vert^2} e(t+1).
\nonumber
\end{flalign}

\begin{prop}
\label{good_prop}
There exists a constant $c_2 $ so that for every $t_0\in\Z$, $ x_0\in\R^{n+m+3d-2}$, $
\theta_0\in{\cal S}$, $ \theta\in{\cal S}_{ab}$, $ r,w\in\ellb_\infty $,
and $\delta\in(0,\infty] $,
when the adaptive controller \eqref{est1} and \eqref{control1}
is applied to the plant \eqref{plant1},
the following holds:
\begin{flalign} 
\overline\phi(t+1)
&=
[\widetilde A_g + \Delta(t) ]
\overline\phi (t)
+
\bar\eta(t),
\quad t\geq t_0 ,
\label{good_key1}
\end{flalign}
with
\begin{flalign}
\|\Delta(t) \|
&\leq
c_2
\sum_{j=2}^{d+1}
\|\nu(t+j) \|,
\label{xx}
\end{flalign}
and 
\begin{flalign}
&\|\bar\eta(t) \|
\leq
c_2
\biggl(
1+
\sum_{j=2}^{d+1}
\|\nu(t+j) \|
\biggr)
\biggl[
\sum_{j=1}^{\max\{3,d+1\}} 
 (|y^*(t+j)|
 +
\nonumber
\\
&\quad
|\overline y^*(t+d+j)|
+
|w(t+j)|+|\overline w(t+j)|)
\biggr].
\label{xx2}
\end{flalign}
\unskip
\end{prop}
\begin{proof}
See the Appendix.
\unskip
\end{proof}
The result in Proposition \ref{good_prop}
looks very similar, though not identical, 
to the analysis leading up to the main result 
of \cite{acc19} and \cite{mcss20}
on the $d$-step ahead adaptive control problem.
Notice that the matrix $\widetilde A_g $ 
is a function of $\theta\in{\cal S}_{ab} $ and 
the coefficients of ${\mathbf L}(z^{-1}) $;
it lies in a corresponding compact set ${\cal A} \subset \R^{(n+m+d+n'(d+1))\times(n+m+d+n'(d+1)) } $.
Furthermore, the eigenvalues of $\widetilde A_g $ are
at the origin, the roots of ${\mathbf L}(z^{-1}) $, and the roots of ${\mathbf B}(z^{-1}) $,
so they are all in the open unit disk;
so
we can use classical arguments to prove that 
for the desired reference model
there exists constants $\gamma $ and $\sigma\in(0 ,1) $ 
so that for all $\theta\in{\cal S}_{ab} $, we have
\begin{equation}
\bigl\|  \widetilde A_g^k  \bigr\| \leq \gamma \sigma^k,\quad k\geq0.
\label{ATV1}
\end{equation}
Indeed, we can choose any $\sigma $ larger than
\[
\underline\lambda:=
\max_{\theta\in{\cal S}_{ab} }
\bigl\{|\lambda|: \lambda\in\C, {\mathbf B}(\lambda^{-1})=0 \text{ and } {\mathbf L}(\lambda^{-1})=0  \bigr\}.
\]

Equations of the form given in \eqref{good_key1} appear
in classical adaptive control approaches.
While we can view \eqref{good_key1} as a linear time-varying system, 
we have to keep in mind that $\Delta(t) $ and $\bar\eta(t) $ are implicit nonlinear functions
of $\theta $, $\theta_0 $, $x_0$, $r $ and $w $.
However, this linear time-varying interpretation is very convenient for analysis;
to this end, let $\boldsymbol\Phi_A $ denote the state transition matrix of a general time-varying
square matrix $A$.
The following result is useful in analyzing our closed-loop system.

\begin{prop}[\hspace{-.2pt}\textbf{\cite{kreiss2}}]
\label{prop_4}
With $\sigma\in(\underline\lambda,1 ) $,
suppose that $\gamma\geq 1 $ is such that \eqref{ATV1} is satisfied for
every $\widetilde A_g \in {\cal A} $.
For every $\mu\in(\sigma,1) $, $g_0\geq0 $, $g_1\geq0 $, and
$g_2 \in \bigl[0, \frac{\mu-\sigma}{\gamma} \bigr),$
there exists a constant $\bar\gamma\geq1 $ so that for every 
$\widetilde A_g\in{\cal A} $ and
$\Delta \in {\mathbb S}\bigl(\R^{(n+m+d+n'(d+1))\times(n+m+d+n'(d+1)) } \bigr) $ satisfying
\[
\sum_{j=\tau}^{t-1}
\|\Delta(j) \|
\leq
g_0 + g_1 (t-\tau)^{\frac{1}{2}}+g_2(t-\tau), \;\; 
\bar t \geq t>\tau \geq \underline t,
\]
we have
$\|{\boldsymbol\Phi}_{\widetilde A_g + \Delta }(t,\tau) \|
\leq
\bar\gamma \mu^{t-\tau},
\quad \bar t \geq t>\tau \geq \underline t.$
\end{prop}

Next,
we present the main result proving that the closed-loop system
enjoys very desirable linear-like behavior.

\section{The Main Result}

\begin{theorem}
\label{thm1}
For every $\delta\in(0,\infty] $ and $\lambda\in(\underline\lambda,1 ) $, 
there exists a constant $c >0 $ so that
for every $t_0\in \Z $, $\theta \in {\cal S}_{ab}  $, $r,w \in \ellb_\infty $,
$ \theta_0 \in {\cal S} $, and plant initial condition
$x_0 $,
when the adaptive controller \eqref{est1} and 
\eqref{control1} 
is applied to the plant \eqref{plant1},
the following bound holds:
\begin{equation}
\|\phi(t)\| 
\leq 
c \lambda^{t-t_0} \|x_0 \|
+
\sum_{j=t_0}^t c \lambda^{t-j }
(|r(j) | + |w(j)| ), \quad t\geq t_0.
\label{th_bd1}
\end{equation}
Furthermore, if $w=0 $, then
\begin{flalign}
\sum_{k=t_0+d}^\infty 
\varepsilon(k)^2
&\leq
c (\|x_0 \|^2 + \|r \|_\infty^2 ).
\label{errBD2}
\end{flalign}
\end{theorem}

\begin{remark}
The above result shows that the closed-loop
system experiences linear-like behavior. There is a uniform
exponential decay bound on the effect of the initial condition,
and a convolution bound on the effect of the exogenous inputs.
This implies that the system has a bounded gain (from $w$ and
$r$
to $y$) in every $p$-norm. For example, for $p = \infty $, we see
from the above bound that
\begin{equation}
\|\phi(t) \|
\leq
\tfrac{c}{1-\lambda} \left(
\lambda^{t-t_0} \|x_0 \|
+ \|w\|_\infty+ \|r\|_\infty,
\right)
\quad
t\geq t_0.
\nonumber
\end{equation}
\end{remark}

\begin{remark}
In the absence of noise, most
adaptive controllers 
provide 
that the tracking error is
square summable, e.g. see \cite{goodwinsin}. 
Here we prove a stronger result, namely, 
an upper bound on the $2$-norm
in terms of the size of $x_0 $ and $r$.
\end{remark}

\begin{proof}[{\bf Proof of Theorem \ref{thm1}}]
Fix $\delta\in(0,\infty] $ and $\lambda\in(\underline\lambda,1) $.
Let $t_0\in \Z $, $\theta \in {\cal S}_{ab}  $, $r,w \in \ellb_\infty $,
$\theta_0 \in {\cal S} $, and $x_0\in\R^{n+m+3d-2} $ be arbitrary.
Now choose $\sigma\in(\underline\lambda,\lambda ) $.
We will analyze \eqref{good_key1} to obtain
a bound on $\overline\phi(t) $.

Before proceeding, we see that there exists $\gamma_1 $ so that for every
$\widetilde A_g \in {\cal A} $,  
$\|\widetilde A_g^k \| \leq \gamma_1 \sigma^k, k\geq0.$
Also, we need to compute a bound on the sum
of $\|\Delta(\cdot ) \| $; since there are $d$ terms 
on the RHS of \eqref{xx}, 
by the Cauchy-Schwarz inequality
we obtain
\begin{flalign}
\sum_{j=\tau}^{t-1} \|\Delta(j) \|
&\leq
d c_2
\sum_{j=\tau+1 }^{t+d-1}
\|\nu(j+1) \|
\nonumber
\\
&\leq
d^2 c_2
\left[
\sum_{j=\tau+1 }^{t+d-1}
\|\nu(j+1) \|^2
\right]^{\frac{1}{2}}
(t-\tau +d-1 )^{\frac{1}{2}},
\nonumber
\\
&\qquad\quad
t>\tau \geq t_0;
\end{flalign}
but $(t_2-t_1+d-1)^\frac{1}{2} \leq d(t_2-t_1)^\frac{1}{2}, t_2>t_1 $, so
incorporating this and the definition of $\nu(\cdot) $ we have
\begin{flalign}
\sum_{j=\tau}^{t-1} \|\Delta(j) \|
&\leq
d^3 c_2
\left[
\sum_{j=\tau+1 }^{t+d-1}
\|\nu(j+1) \|^2
\right]^{\frac{1}{2}}
(t-\tau)^{\frac{1}{2}},
\nonumber
\\
&=
d^3 c_2
\left[
\sum_{j=\tau+2 }^{t+d}
\rho(j)
\frac{|e(j+1)|^2}{\|\phi(j-d+1) \|^2}
\right]^{\frac{1}{2}}
(t-\tau)^{\frac{1}{2}},
\nonumber
\\
&\qquad\quad
t>\tau \geq t_0.
\label{delta_bd1}
\end{flalign}
Also for ease of notation, let us define
\begin{equation}
\tilde w(t)
:=
\sum_{j=1}^{\max\{3,d+1\}} 
(
 |y^*(t+j)|
 +
|\overline y^*(t+d+j)|
+
|w(t+j)|+|\overline w(t+j)|).
\nonumber
\end{equation}

Now we consider the closed-loop system behavior. To proceed, 
we partition the
timeline into two parts: one in which the noise $\overline w(\cdot) $ is small versus 
$\phi(\cdot) $
and one where it is not.
To this end, 
with $\mathfrak{v}>0 $ to be chosen shortly, 
define
\[
S_{\text{\sf good}}:=
\left\{
j\geq t_0 : \phi(j) \neq0 \text{ and } \tfrac{|\overline w(j) |^2 }{\|\phi(j) \|^2 } < \mathfrak{v}
\right\},
\]
\[
S_{\text{\sf bad}}:=
\left\{
j\geq t_0 : \phi(j) =0 \text{ or } \tfrac{|\overline w(j) |^2 }{\|\phi(j) \|^2 } \geq \mathfrak{v}
\right\};
\]
clearly $\{j\in\Z : j\geq t_0 \} =S_{\text{\sf good}} \cup S_{\text{\sf bad}}$.\footnote{If the noise is zero, the $S_{\text{\sf good}}$ may be the whole timeline $[t_0,\infty) $.}
Observe that this partition implicitly 
depends on $\theta\in{\cal S}_{ab} $, as well as the initial
conditions. We will easily obtain bounds on the closed-loop system behavior on $S_{\text{\sf bad}}$; 
we will apply Proposition \ref{prop_4} to analyze the behavior on $S_{\text{\sf good}}$.
Before proceeding, we partition the timeline
into intervals which oscillate between $S_{\text{\sf good}}$
and $S_{\text{\sf bad}}$. To this end, it is easy to see
that we can define a (possibly infinite) sequence of intervals of the form $[k_i,k_{i+1}) $ satisfying:
(i) $k_0=t_0 $; (ii) $[k_i,k_{i+1}) $ either belongs to
$S_{\text{\sf good}}$ or $S_{\text{\sf bad}}$; and
(iii) if $k_{i+1}\neq \infty $ and $[k_i,k_{i+1}) $ belongs
to $S_{\text{\sf good}}$ (respectively, $S_{\text{\sf bad}}$), then the interval $[k_{i+1},k_{i+2}) $ must belong to
$S_{\text{\sf bad}}$ (respectively, $S_{\text{\sf good}}$).

Now we analyze the closed-loop behavior on each interval.

\noindent{\bf Case 1}: 
The behavior on $S_{\text{\sf bad}}$.

Let $j\in [k_i,k_{i+1})\subset S_{\text{\sf bad}} $ be arbitrary. In this case,
we have
either $\phi(j)=0 $ or $\tfrac{|\overline w(j) |^2}{\|\phi(j) \|^2 } \geq \mathfrak{v} $. 
In either case, we have
\begin{flalign}
\|\phi(j) \|\leq \tfrac{1}{\sqrt{\mathfrak{v}} } 
|\overline w(j) |,
\quad j\in[k_i,k_{i+1} );
\label{bad_bd1}
\end{flalign}
then from the crude model \eqref{crude1} 
and Proposition \ref{prop_crude},
we have
\begin{equation}
\|\phi(j+1) \|
\leq
\tfrac{c_1}{\sqrt{\mathfrak{v}} }
|\overline w(j) | 
+
c_1 |\overline y^*(j+d+1) |
+
c_1 |w(j+1) |,
\;\;
j\in[k_i,k_{i+1});
\nonumber
\end{equation}
combining this with \eqref{bad_bd1} yields:
\begin{equation}
\|\phi(j) \|
\leq 
\left\lbrace
\begin{matrix*}[l]
\tfrac{1}{\sqrt{\mathfrak{v}} }
|\overline w(j) |,
& j=k_i
\\ 
c_1 \left(\tfrac{1}{\sqrt{\mathfrak{v}} }+1\right)
[|\overline w(j-1) | 
+
\\
\quad
|\overline y^*(j+d) |
+
|w(j) |],
&j=k_i+1,\ldots,k_{i+1}.
\end{matrix*}
\right.
\nonumber
\end{equation}

\noindent{\bf Case 2}: 
The behavior on $S_{\text{\sf good}}$.

Suppose that $[k_i,k_{i+1} ) $ lies in $S_{\text{\sf good}}$;
notice that the bound on $\|\Delta(t) \| $ in \eqref{delta_bd1} 
occasionally extends outside $S_{\text{\sf good}}$;
so we handle the first $d+1$ and last $d+1$ times steps separately.

To this end, first suppose that $k_{i+1}-k_i\leq 2(d+1) $; 
then using the crude model on $\phi $ in \eqref{crude1} and Proposition \ref{prop_crude}, it is easy to show that if we define $\gamma_2:= \left(\frac{c_1}{\lambda}\right)^{2d+2} $, then we have
\begin{flalign} 
&\|\phi(t) \| \leq
\gamma_2 \lambda^{t-k_i} \|\phi(k_i) \|
+
\nonumber
\\
&
\sum_{j=k_i}^{t-1}
\gamma_2 \lambda^{t-j-1}
(|\overline y^*(j+d+1)|+|w(j+1)| ),
\;\; 
t\in[k_i,k_{i+1}].
\label{good_bd4}
\end{flalign}

Now suppose that $k_{i+1}-k_i> 2(d+1) $.
Define $\overline k_i:= k_i+d+1 $ and $\underline k_i:= k_{i+1}-d-1 $.
By the second part of Proposition \ref{est_prop} and using the facts 
that $\|\tilde \theta(t) \|\leq 2\|{\cal S} \|  $
and that 
$\frac{|\overline w(j)|^2 }{\|\phi(j) \|^2}< \mathfrak{v}   $ for $j\in[k_i,k_{i+1} ) $,
we obtain from \eqref{delta_bd1}:
\begin{flalign}
\sum_{j=\tau}^{t-1}
\|\Delta(j)\|
&\leq 
d^3c_2
\left[
8\|{\cal S} \|^2+
4\mathfrak{v}(t-\tau+d-1)
\right]^\frac{1}{2}
(t-\tau)^\frac{1}{2} ,
\nonumber
\\
&
\qquad \underline k_{i+1} \geq t>\tau \geq \overline k_i.
\nonumber
\end{flalign}
If we restrict $\mathfrak{v}\leq1 $, 
and define $\gamma_3:= d^3c_2\bigl(
\bigl[
8\|{\cal S} \|^2+
4(d-1)
\bigr]^\frac{1}{2} + 2
\bigr)
 $, 
then we obtain
\begin{equation}
\sum_{j=\tau}^{t-1}
\|\Delta(j)\|
\leq 
\gamma_3
(t-\tau)^\frac{1}{2}
+
\gamma_3
\mathfrak{v}^\frac{1}{2}
(t-\tau),
\quad \underline k_{i+1} \geq t>\tau \geq \overline k_i.
\nonumber
\end{equation}
We now apply Proposition \ref{prop_4}: 
set $g_0=0, g_1=\gamma_3, g_2=\gamma_3\mathfrak{v}^\frac{1}{2}, \mu=\lambda, \gamma=\gamma_1 $;
we need $\gamma_3\mathfrak{v}^\frac{1}{2}< \tfrac{\lambda-\sigma}{\gamma_1} $,
so if we set $\mathfrak{v}:= \min \left\{1, 
\tfrac{1}{2} \bigl(\tfrac{\lambda-\sigma}{\gamma_3\gamma_1}\bigr)^2 \right\} $,
then from Proposition \ref{prop_4}
we see that there exists a constant $\gamma_4 $ so that
the state transition 
matrix $\boldsymbol\Phi_{\widetilde A_g+\Delta }(t,\tau)  $ satisfies
\begin{equation}
\|\boldsymbol\Phi_{\widetilde A_g+\Delta }(t,\tau)\|
\leq
\gamma_4 \lambda^{t-\tau},
\quad \underline k_{i+1} \geq t>\tau \geq \overline k_i.
\label{state_bd1}
\end{equation}
Before solving \eqref{good_key1}, we obtain a bound on $\bar\eta(t) $;
from Proposition \ref{est_prop}, 
we see that 
$\|\nu(t)\| \leq \sqrt{8\|{\cal S} \|^2+4\mathfrak{v} }
,\underline k_{i+1} \geq t \geq \overline k_i  $,
so there exists a constant $\gamma_5 $ so that
$\|\bar\eta(t) \|
\leq
\gamma_5
\tilde w(t),
\;
\underline k_{i+1} \geq t \geq \overline k_i.$
Then, using the bound in \eqref{state_bd1} to
solve \eqref{good_key1} we see that there exists a constant $\gamma_6 $ 
so that
\begin{equation} 
\|\overline \phi(t) \| 
\leq
\gamma_6 \lambda^{t-k_i} \|\overline \phi(\overline k_i) \|
+
\sum_{j=k_i}^{t-1}
\gamma_6 \lambda^{t-j-1}
\tilde w(j),
\quad 
t\in[\overline k_i,\underline k_{i+1}].
\label{good_bd5}
\end{equation}
We want to have a bound on the whole interval $[k_i,k_{i+1} ) $,
and we would like it to be in terms of $\phi $ instead of $\overline\phi $.
First, we
use the crude model on $\overline\phi(\cdot) $ in \eqref{crude_no3} 
and Proposition \ref{prop_crude} to find bounds 
on $\|\overline\phi(t) \| $ for $t\in[k_i, \overline k_i  ] $ and
for $t\in[\underline k_{i+1}, k_{i+1}  ] $, 
and on $\|\overline\phi(\overline k_i) \|$ in terms of $\|\overline\phi(k_i) \|$  
and combine them 
with \eqref{good_bd5} 
to see that there exists a constant $\gamma_7 $
so that
\begin{equation} 
\|\overline \phi(t) \| 
\leq
\gamma_7 \lambda^{t-k_i} \|\overline \phi(k_i) \|
+
\sum_{j=k_i}^{t-1}
\gamma_7 \lambda^{t-j-1}
\tilde w(j),
\quad 
t\in[ k_i, k_{i+1}].
\label{good_bd6}
\end{equation}
Next we obtain a bound in terms of $\phi$.
It is obvious that $ \|\phi(t) \|\leq\|\overline\phi(t) \|,t\in[ k_i, k_{i+1}]  $.
Then from the definitions of $\overline\phi(\cdot),\overline\zeta(\cdot)  $ and 
$\zeta(\cdot) $, it is easy to see that there exists a constant $c_3 $ such that
$\|\overline \phi(k_i) \|
\leq
c_3
\sum_{j=1}^{d+1}
\|\phi(k_i+j) \|
+
c_3
\sum_{j=-n'+2}^{d+1} |y^*(k_i+j)|;$
we
use the crude model on $\phi(\cdot) $ in \eqref{crude1} and Proposition \ref{prop_crude} 
to obtain bounds on $\|\phi(k_i+j) \|,j=1,2,\ldots,d+1 $, 
in terms of $\|\phi(k_i) \|$.
Incorporating all of the above into \eqref{good_bd6} and 
after simplification, we
see that there exists a constant $\gamma_8 $ so that
\begin{flalign} 
&\| \phi(t) \| 
\leq
\gamma_8 \lambda^{t-k_i} \| \phi(k_i) \|
+
\nonumber
\\
&
\sum_{j=k_i}^{t-1}
\gamma_8 \lambda^{t-j-1}
\tilde w(j)+
\gamma_8
\sum_{q=1}^{n'-2} |y^*(k_i-q)|,
\quad 
t\in[ k_i, k_{i+1}],
\nonumber
\end{flalign}
which we combine with \eqref{good_bd4} 
to conclude Case 2.

We now glue together the bounds obtained on $S_{\text{\sf good}}$ and $S_{\text{\sf bad}}$
to
obtain a bound which holds on all of $[t_0,\infty )$
using the identical argument used in gluing together similar 
bounds in the proof of Theorem 1 of \cite{mcss20}.
Last of all, we simplify the resulting 
quantity by using causality arguments
to remove extraneous terms, and end up with
the bound in \eqref{th_bd1}.

Finally we prove asymptotic tracking.
Suppose that $w=0$; then we see from the estimation algorithm
that in this case, $\rho(t)=1 \Leftrightarrow \|\phi(t-d) \|\neq0 $.
So from \eqref{error1}, the first part of Proposition \ref{est_prop}, and the Cauchy-Schwarz inequality,
we have 
\[
\rho(t-1)
\frac{\overline \varepsilon(t)^2}{\|\phi(t-d) \|^2}
\leq
d \sum_{j=0,\|\phi(t-j-d) \|\neq0}^{d-1}
\frac{e(t-j)^2}{\|\phi(t-j-d) \|^2};
\]
so by the second part of Proposition \ref{est_prop} 
we can see that
\begin{flalign}
&\sum_{t=t_0+d,\|\phi(t-d) \|\neq0 }^\infty
\frac{\overline \varepsilon(t)^2}{\|\phi(t-d) \|^2}
\leq
8d^2\|{\cal S}\|^2.
\nonumber
\end{flalign}
Then it is easy to see
by the boundedness of $\phi $ proven in \eqref{th_bd1}
and by the fact that $\overline \varepsilon(t)=0 $ when $\phi(t-d)=0 $,
 that  
 \begin{equation}
\sum_{t=t_0+d}^\infty
\overline \varepsilon(t)^2
\leq
8d^2\|{\cal S}\|^2\max_{j\geq t_0-d} \|\phi(j) \|^2
\leq
\left(\tfrac{4d\|{\cal S} \|c}{1-\lambda}\right)^2
[ \|x_0 \|^2+ \|r\|_\infty^2].
\nonumber
\end{equation}
But $\varepsilon $ and $\overline\varepsilon $
are related by a stable transfer function 
(see \eqref{tfL}),
so if we apply Parseval's Theorem then we 
obtain a bound on $\varepsilon $ of the form
given in \eqref{errBD2}.
\end{proof}

\section{Robustness}

It turns out that the convolution bounds
proven in Theorem \ref{thm1} will guarantee robustness to
a degree of time-variations and unmodelled dynamics.
We have shown that
the corresponding model reference adaptive controller provides
a convolution bound with gain $c$
 and decay rate $\lambda$ when
applied to the time-invariant nominal plant; we can apply
Theorems 1 and 2 of \cite{ccta20} to
show that, in the presence of a degree of time-variation
(slow enough parameter time-variations and/or occasional
jumps) and small enough unmodelled dynamics,
the controller still provides linear-like properties.
Furthermore, we can also obtain 
bounds on the average tracking error both
in the case of no noise under slow time-variations,
as well as in the noisy case, by adapting 
arguments used
in the proofs of Theorems 4 and 5 of \cite{mcss20}.

\section{A Simulation Example}
\unskip
We now provide a simulation example to illustrate the
results of this paper. Consider the time-varying plant
\begin{multline*}
y(t + 1) =
 -a_1(t)y(t) - a_2(t)y(t - 1)
+ 
\\
b_0(t)u(t) + b_1(t)u(t - 1) + w(t),
\end{multline*}
with $a_1(t) \in [-2, 2], a_2(t) \in [-2, 2], b_0(t) \in [\tfrac{3}{2}, 5]$ and
$b_1(t) \in [-1, 1]$. 
Note here that the delay $d=1$.
We want to apply an adaptive controller such that the closed-loop
system follows the behavior of a reference model \eqref{plantRef1} with $n'=2 $;
following the discussion at the beginning of Section \ref{sec2},
we transform the plant into the predictor form by way of long division:
we see that
$\alpha_0(t)=a_1(t)-l_1$,
$ \alpha_1(t)=a_2(t)-l_2$, $\beta_0(t)=b_0(t)$, and $\beta_1(t)=b_1(t)$.
We choose a reference model 
represented by
${\mathbf L}(z^{-1}):=1-\frac{1}{2}z^{-2},
$ and $
{\mathbf H}(z^{-1}):=\frac{1}{2},$
which has poles in the open unit disk 
as required;
then we can set
\begin{multline*}
{\cal S}:=
{\cal S}_{\alpha\beta}=
\biggl\{
\left[
\begin{smallmatrix}
\alpha_0\\\alpha_1\\\beta_0\\\beta_1
\end{smallmatrix}
\right] \in\R^4
:
\alpha_0\in[-2,2], \alpha_1\in[-\tfrac{5}{2},\tfrac{3}{2}],
\\
 \beta_0\in [\tfrac{3}{2}, 5], \beta_1\in [-1, 1]
\biggr\}.
\end{multline*}
We apply the adaptive controller \eqref{est1} and \eqref{control1}
(with $\delta =\infty $) 
to this plant 
with the plant parameters given by:
$a_1(t) 
= 2 \cos(\tfrac{1}{100}t)$,
$a_2(t) = -2 \sin(\tfrac{1}{300}t)$,
$b_0(t) 
= \tfrac{13}{4} - \tfrac{7}{4} \cos(\tfrac{1}{125}t)$,
and
$b_1(t) = -\cos(\tfrac{1}{50}t)$,
and the disturbance given by:
$$
w(t) =
\left\{
\begin{smallmatrix*}[l]
\tfrac{1}{10} \cos(10t), & 200 < t \leq 500
\\
0, & \text{otherwise}.
\end{smallmatrix*}
\right.
.$$
We set
$r(t) $ to be a unit square wave with period of $200$ steps.
We set $y(-1) = y(0) = -1, u(-1) = 0$, and the initial parameter
estimates to the midpoint of the respective intervals.
Figures \ref{fig1} and \ref{fig2} show the results. 
The controller does
a good job of tracking when there is no disturbance;
the tracking degrades when the disturbance enters the system 
but
tracking performance improves when the disturbance returns
to zero. 
You can also see that the estimator tracks the
time-varying parameters fairly well.

\begin{figure}
\center\includegraphics[trim=30 10 30 20,clip,width=.73\columnwidth]{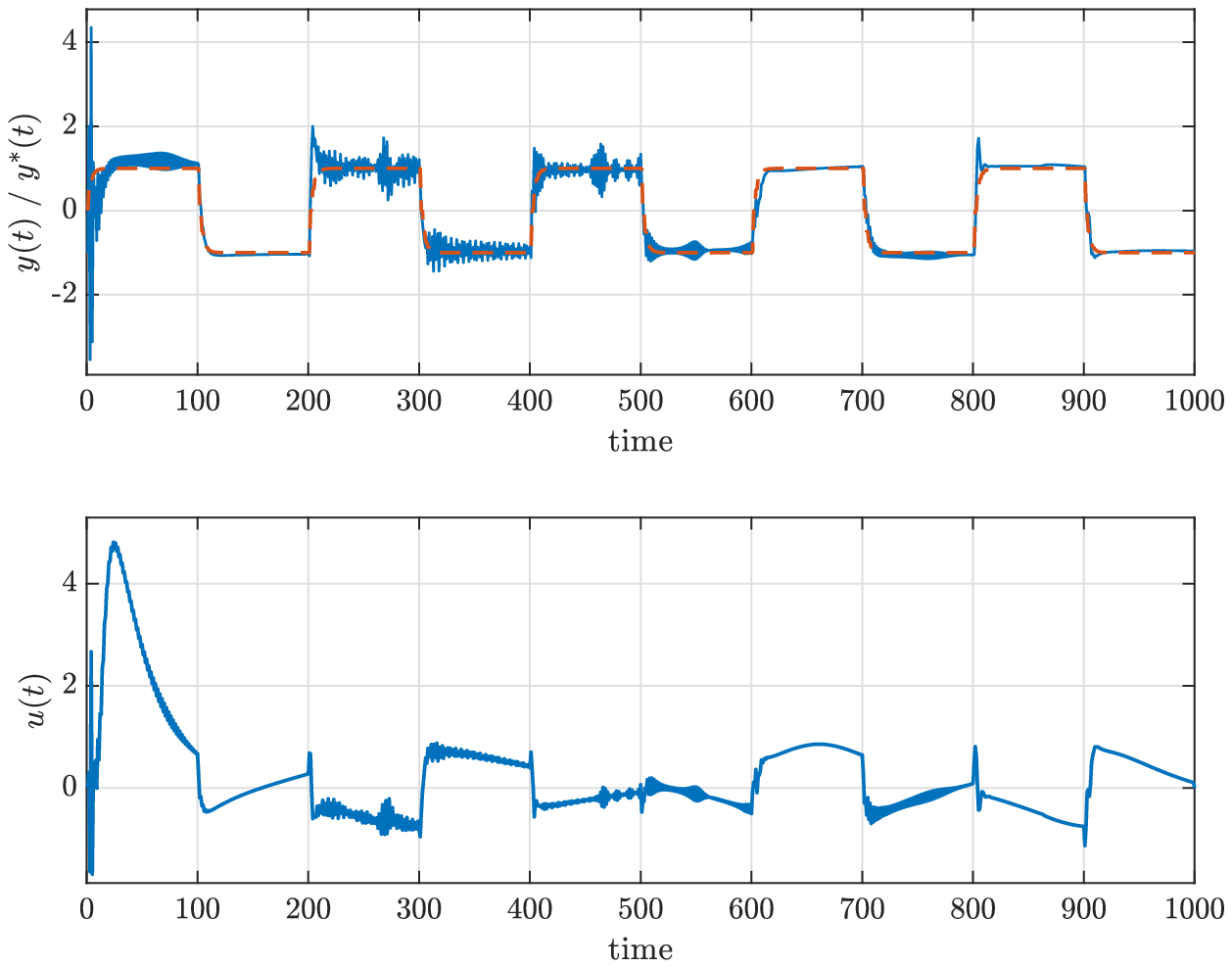}
\vspace{-1em}
\caption{
The first plot shows both $y(t)$ (solid) and $y^*(t) $ (dashed);
the second plot shows the control input $u(t)$.
}
\vspace{-1.2em}
\label{fig1}
\end{figure}
\begin{figure}
\center\includegraphics[trim=30 20 30 20,clip,width=.73\columnwidth]{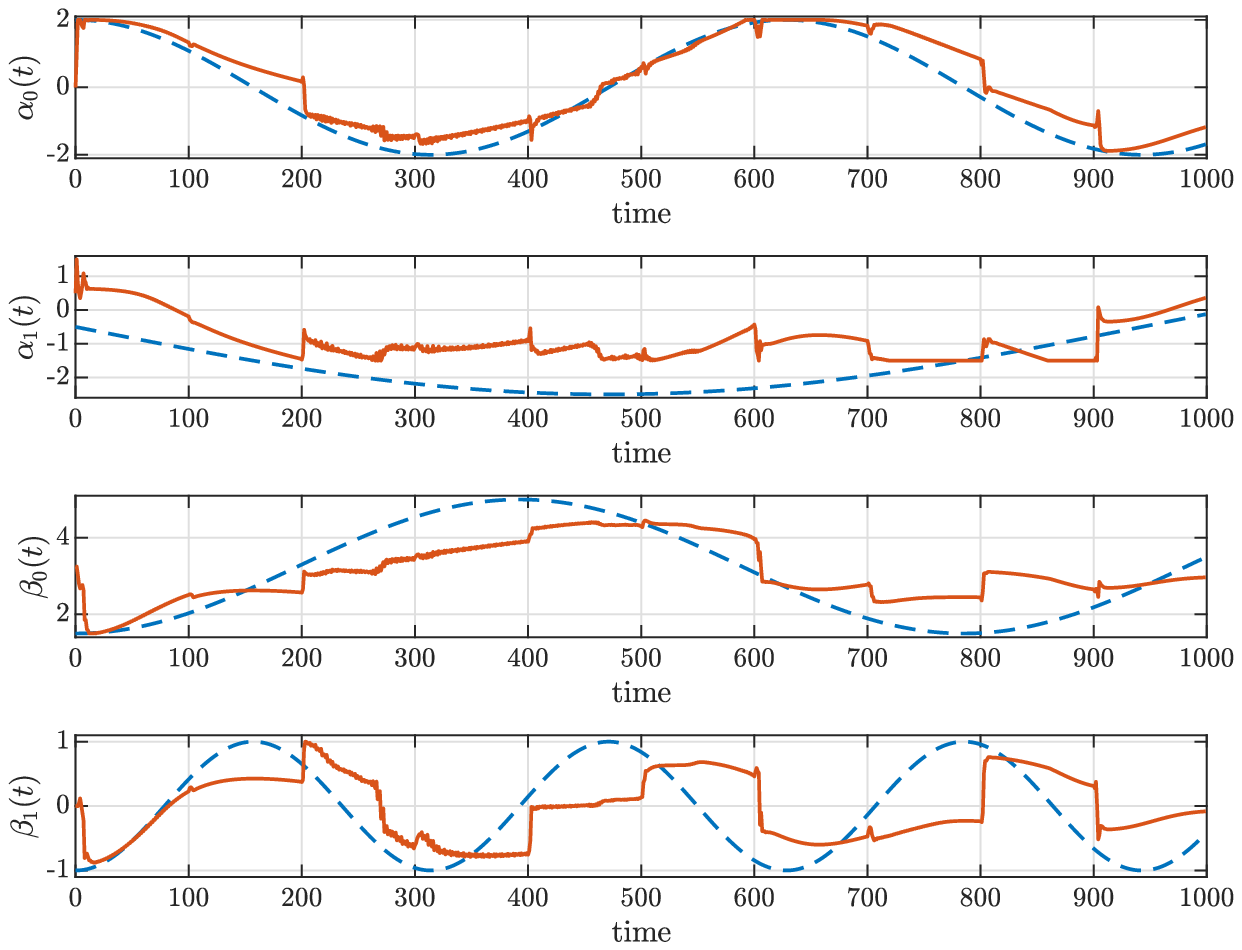}
\vspace{-1.5em}
\caption{
The plots show the parameter estimates $\hat\theta(t) $ (solid)
as well as the actual parameters $\theta^* $ (dashed).
}
\vspace{-1.5em}
\label{fig2}
\end{figure}

\vspace{-.7em}

\section{Conclusion}
\unskip
In this paper,
we show that
a model reference adaptive controller 
using a parameter estimator based on the original 
projection algorithm 
provides desirable {\bf linear-like closed-loop properties}:
exponential stability, a bounded gain on the noise in every $p$-norm,
and a convolution bound on the exogenous inputs;
this is never found in the literature,
except in our earlier work.
This can be leveraged to directly 
prove tolerance to a degree of parameter time-variations and
a degree of unmodelled dynamics.
Also, in the noise-free case, asymptotic tracking is achieved
with an explicit upper bound on the square sum of the tracking error.

We would like to extend these linear-like results 
to the case when the sign of the high-frequency gain is unknown
by using multiple estimators along the lines 
of \cite{cdc18} and \cite{tac20}.
At present, our proof does not extend to that case
since the key transition matrix is not deadbeat,
but we are working on an alternative proof approach.

\vspace*{-.3cm}

\appendix
\unskip
\begin{proof}[\bf Proof of Proposition \ref{good_prop}]
\unskip
First, define a square matrix 
  of size $(n+m+d+n'(d+1))$,
$P:=
\left[
\begin{matrix}
I_{n+m+d} & 
\\
 & {\boldsymbol 0}
\end{matrix}
\right];$
then we have
\begin{equation}
\overline\phi(t)^\top P=\begin{bmatrix}
\phi(t)^\top & {\boldsymbol 0} \end{bmatrix},
\label{Peq1}
\end{equation}
so
$\overline\phi(t)^\top P \overline\phi(t) = \phi(t)^\top \phi(t)= \|\phi(t)\|^2 .$
Now define
\begin{equation}
	\tilde\Delta(t):=\rho(t-1)
	\frac{\overline\varepsilon(t)}{\Vert\phi(t-d)\Vert^2} 
	\e_{n+m+d+1}
	\overline \phi(t-d)^\top P
	,\;\;
	t\geq t_0+1,
	\label{deltaj_def}
  \end{equation}
and $\eta_0(t):= [1-\rho(t-1)]\overline \varepsilon(t), t\geq t_0 +1$;
then by using these definitions, we can represent the 
term
containing
$\overline\varepsilon(t+d+2)$
in the RHS of \eqref{goodmodel3}
as: 
\begin{flalign}
\e_{n+m+d+1} \overline\varepsilon(t+d+2)
&=
\tilde\Delta(t+d+2) \overline \phi(t+2) 
	+ 
\nonumber
\\
&\qquad
	\e_{n+m+d+1} \eta_0(t+d+2) .	
	\label{vare_terms1}	
\end{flalign}
We now use \eqref{crude_no3} (which is valid for $t\geq t_0 $)
to represent $\overline\phi(t+2) $
in the RHS of \eqref{vare_terms1}
in terms $\overline\phi(t) $;
so if we do this and incorporate the result into \eqref{vare_terms1},
then we are now able to rewrite the model \eqref{goodmodel3}
in the desired form of \eqref{good_key1}
by defining
\begin{flalign}
\label{deltabig_def}
\Delta(t)
&:=
\tilde\Delta(t+d+2)
A_2(t+1)A_2(t)
,
\quad t\geq t_0,
\end{flalign}
and by grouping 
the remaining terms (containing exogenous signals) and
defining
\begin{flalign}
&\bar\eta(t)
:=
\eta(t)
+
\e_{n+m+d+1} \eta_0(t+d+2)
+
\nonumber
\\
&\quad
\tilde\Delta(t+d+2)
\biggl[
A_2(t+1)B_5(t)\overline y^*(t+d+1)+
\nonumber
\\
&\quad
A_2(t+1)B_6(t)w(t+1)+
\nonumber
\\
&\quad
\bigl(A_2(t+1)B_7(t)+
B_5(t+1)\bigr)\overline y^*(t+d+2)
+
\nonumber
\\
&\quad
\bigl(A_2(t+1)B_8(t)+B_6(t+1)\bigr)w(t+2)
+
\nonumber
\\
& \quad
A_2(t+1) \e_{n+m+d+1} \overline w(t+2)+
B_7(t+1)\overline y^*(t+d+3)
+
\nonumber
\\
&\quad
B_8(t+1)w(t+3)+
\e_{n+m+d+1} \overline w(t+3)
\biggr],
\label{eta_def1}
\end{flalign}
concluding the first part of the proof.

We now prove the desired bound on $\Delta(t) $. 
By \eqref{Peq1}
we obtain from \eqref{deltaj_def}: 
$\|\tilde\Delta(t) \| \leq \rho(t-1)
\frac{|\overline\varepsilon(t)|}{\Vert\phi(t-d)\Vert},
\; t\geq t_0+1; $
from \eqref{error1} and the first part of Proposition \ref{est_prop}, we can see that
\begin{flalign}
&\|\tilde\Delta(t) \| \leq
\rho(t-1)
\frac{|\overline\varepsilon(t)|}{\Vert\phi(t-d)\Vert}
\nonumber
\\
&\leq 
\rho(t-1)
\frac{|e(t)|}{\Vert\phi(t-d)\Vert}
+
\|\hat\theta(t-1)-\hat\theta(t-d) \|
\nonumber
\\
&\leq 
\sum_{j=0}^{d-1}
\rho(t-1-j)
\frac{|e(t-j)|}{\Vert\phi(t-d-j)\Vert}
=
\sum_{j=1}^{d}
\|\nu(t-j)\|,
\nonumber
\\
&
\quad 
t\geq t_0+d
.
\label{prede_bound2}
\end{flalign}
Then from \eqref{deltabig_def},  
using the bound in \eqref{prede_bound2} and Proposition 2, 
we can easily show that there exist a constant so that \eqref{xx} is proven.

Finally, we prove the bound on $\bar\eta $.
From \eqref{error2} and the definition of $\rho(\cdot) $,
it is easy to show that
$|\eta_0(t)|
\leq 
\bigl(
1+ \tfrac{4\|{\cal S}\| }{\delta}
\bigr)
|\overline w(t-d) |,
\;
t\geq t_0 +d ;$
then 
by incorporating this bound 
and using the definition of $\eta(t) $ in \eqref{eta_2}
along with the bound in \eqref{xx}
into \eqref{eta_def1},
it is easy to see that there exists a constant
so that we obtain the desired bound \eqref{xx2}.
\end{proof}

\vspace*{-1.5em}

\bibliographystyle{IEEEtranSmod_v2}
\bibliography{mrac_refs_arxiv}           

\end{document}